\documentclass[sigconf]{acmart}
\usepackage{multirow}
\usepackage{float}
\usepackage{enumitem}
\usepackage{amsmath}
\usepackage{algorithmic}
\newtheorem{assumption}{Assumption}
\newtheorem{theorem}{Theorem}

\usepackage[ruled, linesnumbered]{algorithm2e}
\usepackage[table]{xcolor}
\usepackage{colortbl}
\usepackage{subcaption}
\usepackage{pifont}
\theoremstyle{remark}

\DeclareMathOperator*{\argmin}{arg\,min}

\begin{document}

\title{CLEAR: Null-Space Projection for Cross-Modal De-Redundancy in Multimodal Recommendation}

\author{Hao Zhan}
\authornote{These authors contributed equally to this work.}
\affiliation{
  \institution{Hefei University of Technology}
  \country{}
}
\email{2023213432@mail.hfut.edu.cn}

\author{Yihui Wang}
\authornotemark[1]
\affiliation{
  \institution{Hefei University of Technology}
  \country{}
}
\email{2023213406@mail.hfut.edu.cn}

\author{Yonghui Yang}
\authornote{Corresponding author.}
\affiliation{
  \institution{National University of Singapore}
  \country{}
}
\email{yh_yang@nus.edu.sg}




\author{Danyang Yue}
\affiliation{
\institution{Huazhong University of Science and Technology}
\city{}
\country{}
}
\email{danyangy@hust.edu.cn}

\author{Yu Wang}
\affiliation{
\institution{Hefei University of Technology}
\city{}
\country{}
}
\email{wangyu20001162@gmail.com}

\author{Pengyang Shao}
\affiliation{
\institution{National University of Singapore}
\city{}
\country{}
}
\email{shaopymark@gmail.com}

\author{Fei Shen}
\affiliation{
\institution{National University of Singapore}
\city{}
\country{}
}
\email{shenfei29@nus.edu.sg}

\author{Fei Liu}
\affiliation{
\institution{Hefei University of Technology}
\city{}
\country{}
}
\email{feiliu@mail.hfut.edu.cn}

\author{Le Wu}
\affiliation{
\institution{Hefei University of Technology}
\city{}
\country{}
}
\email{lewu.ustc@gmail.com}


\renewcommand{\shortauthors}{Ho Zhan et al.}

\newcommand{\fullname}{\textit{\textbf{D}e-\textbf{R}edundancy oriented \textbf{M}ultimodal \textbf{R}epresentation~(\textbf{DRMR})}}
\newcommand{\shortname}{DRMR}

\begin{abstract}
Multimodal recommendation has emerged as an effective paradigm for enhancing collaborative filtering by incorporating heterogeneous content modalities. Existing multimodal recommenders predominantly focus on reinforcing cross-modal consistency to facilitate multimodal fusion. However, we observe that multimodal representations often exhibit substantial cross-modal redundancy, where dominant shared components overlap across modalities. Such redundancy can limit the effective utilization of complementary information, explaining why incorporating additional modalities does not always yield performance improvements. In this work, we propose CLEAR, a lightweight and plug-and-play cross-modal de-redundancy approach for multimodal recommendation. Rather than enforcing stronger cross-modal alignment, CLEAR explicitly characterizes the redundant shared subspace across modalities by modeling cross-modal covariance between visual and textual representations. By identifying dominant shared directions via singular value decomposition and projecting multimodal features onto the complementary null space, CLEAR reshapes the multimodal representation space by suppressing redundant cross-modal components while preserving modality-specific information. This subspace-level projection implicitly regulates representation learning dynamics, preventing the model from repeatedly amplifying redundant shared semantics during training. Notably, CLEAR can be seamlessly integrated into existing multimodal recommenders without modifying their architectures or training objectives. Extensive experiments on three public benchmark datasets demonstrate that explicitly reducing cross-modal redundancy consistently improves recommendation performance across a wide range of multimodal recommendation models.
\end{abstract}



\maketitle

\section{Introduction}
Multimodal recommendation has emerged as a widely adopted paradigm for enhancing collaborative filtering by incorporating auxiliary information from multiple modalities~\cite{he2016vbpr, liu2024multimodal}. By leveraging heterogeneous item content beyond user–item interaction data, multimodal recommenders aim to learn richer and more expressive item representations. Existing multimodal recommendation approaches can be broadly categorized into two lines of research: enhancing item representations with semantic features extracted from individual modalities~\cite{he2016vbpr, chen2017ACF}; and exploiting latent item–item structures induced by modality-specific similarities~\cite{zhang2021LATTICE}. Building upon these modeling paradigms, a growing body of recent work further introduces alignment-driven objectives to improve multimodal fusion, typically by enforcing cross-modal consistency within a shared latent space~\cite{MICOR, liu2023multimodal}.

Despite their effectiveness, existing multimodal recommendation models often implicitly assume that different modalities provide complementary information. In practice, however, auxiliary modalities may exhibit substantial semantic overlap, resulting in significant cross-modal redundancy~\cite{zhang2023mining, wang2025multimodal}. Such redundancy not only diminishes the marginal gains from incorporating additional modalities, but also introduces noise and increases the risk of overfitting when redundant signals are indiscriminately fused. Moreover, alignment-driven objectives, while effective in promoting cross-modal consistency, may inadvertently exacerbate this issue by encouraging different modalities to collapse into similar latent embeddings~\cite{chaudhuri2025closer, yi2025decipher}. Consequently, the learned multimodal representations may underutilize modality-specific characteristics, thereby limiting the full potential of multimodal recommendation.

To better understand the nature of cross-modal redundancy, we conduct simple empirical analyses to illustrate the prevalence of cross-modal redundancy in real-world recommendation data. As shown in Figure~\ref{fig: intro_redundancy}, given an anchor item, we retrieve similar items separately based on pre-trained visual and textual features. We find that the similarity score distributions of the retrieved items under the two modalities exhibit substantial overlap, indicating that visual- and text-based retrieval induce highly similar item neighborhoods in the representation space. Moreover, the overlap ratio between the top-K retrieved items from visual and textual modalities is consistently an order of magnitude higher than random expectation. These observations indicate that different modalities may give rise to highly redundant item–item structures even before being jointly modeled. In practice, such redundancy often manifests as performance saturation in multimodal recommenders: simply introducing additional content modalities on top of already informative representations does not necessarily improve recommendation performances~\cite{yang2025less, ye2025multimodal}. Therefore, simply aggregating multiple modalities is insufficient; the key challenge is to effectively exploit complementary modality-specific information while suppressing redundant shared semantics.



To tackle the redundancy challenge, we propose \textbf{C}ross-moda\textbf{L} de-r\textbf{E}dundant subsp\textbf{A}ce p\textbf{R}ojection~(\textsc{CLEAR}), a lightweight and plug-and-play framework for suppressing cross-modal redundancy in multimodal recommendation. \textit{The core idea of CLEAR is to treat cross-modal redundancy as a low-dimensional shared subspace}: when visual and textual modalities encode overlapping semantics, their representations tend to correlate along a few dominant directions, which are repeatedly reinforced during training and thus dilute modality-specific signals. To operationalize this idea, CLEAR first explicitly models the cross-modal covariance between visual and textual representations and then applies singular value decomposition to extract the dominant shared directions. Rather than enforcing stronger alignment or introducing auxiliary objectives, CLEAR performs a subspace-level intervention by projecting multimodal representations onto the null-space orthogonal to these shared directions, thereby preventing redundant components from being reinforced during training. This projection reshapes the multimodal representation space, mitigating the amplification of shared semantics while preserving complementary information.

\begin{figure} 
    \centering
    \includegraphics[width=0.99\linewidth]
    {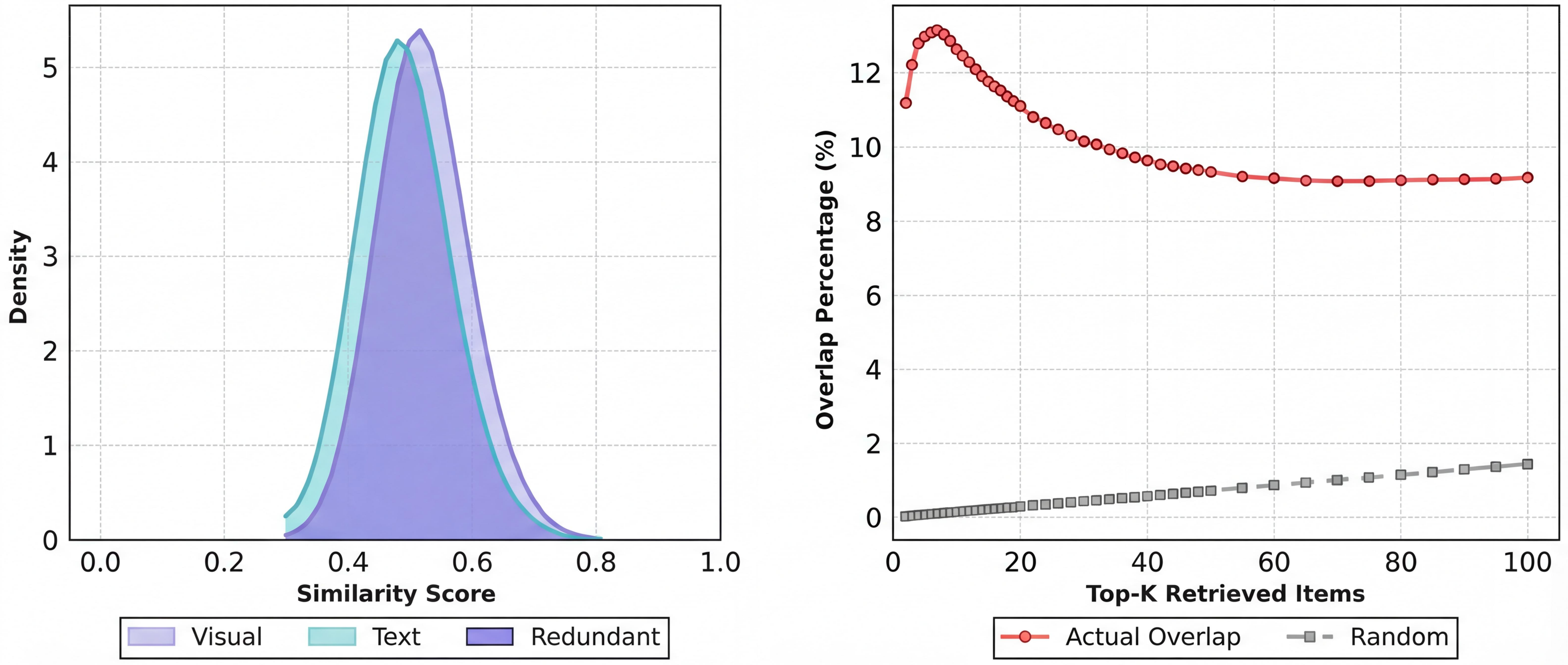}
    \vspace{-0.2cm}
    \caption{Illustration of cross-modal redundancy. Left: density distributions of similarity scores for items retrieved by visual and textual modalities. Right: overlap ratios between the top-K retrieved items across modalities.}
    \label{fig: intro_redundancy}
    \vspace{-0.2cm}
\end{figure}

CLEAR further provides two intuitive control factors to balance redundancy suppression and information preservation: the redundancy rank $k$, which specifies how many dominant shared directions are treated as redundant, and the projection strength $\lambda$, which controls the degree of suppression. By introducing lightweight control over redundancy at the representation level, CLEAR intervenes during training by applying null-space projection to multimodal representations, requiring no architectural modifications or auxiliary loss terms. Owing to its simplicity and generality, CLEAR can be seamlessly integrated into existing multimodal recommenders as a plug-and-play module. Extensive experiments on three benchmark datasets demonstrate that explicitly reducing cross-modal redundancy consistently improves recommendation performance across diverse multimodal recommendation architectures. Our contributions are summarized as follows:
\begin{itemize}[leftmargin=2.0em]
    \item We reframe multimodal recommendation by highlighting that its effectiveness depends on \emph{principled redundancy control across modalities}, instead of increasing modeling capacity or enforcing stronger cross-modal alignment.
    \item We propose \textsc{CLEAR}, a null-space projection framework that suppresses cross-modal redundancy by intervening at the representation level through projection on multimodal features, enabling controlled redundancy suppression while preserving modality-specific information.
    \item We conduct extensive experiments on three benchmark datasets, demonstrating that explicitly reducing cross-modal redundancy with \textsc{CLEAR} consistently improves recommendation performance across diverse multimodal recommendation architectures.
\end{itemize}

\section{Preliminaries}

\subsection{Problem Definition}

In multimodal recommendation, we aim to improve preference prediction by leveraging user-item interactions alongside item multimodal features. Formally, let $\mathcal{U} = \{u_1, u_2, \dots, u_M\}$ and $\mathcal{I} = \{i_1, i_2, \dots, i_N\}$ denote the sets of users and items. The observed interactions form a binary matrix $\mathbf{R} \in \{0,1\}^{M \times N}$, where $r_{ui}=1$ indicates an interaction between user $u$ and item $i$. Each item $i$ is associated with raw visual features $\mathbf{v}_i^{\text{raw}} \in \mathbb{R}^{d_v}$ and textual features $\mathbf{t}_i^{\text{raw}} \in \mathbb{R}^{d_t}$, collectively denoted as matrices $\mathbf{V}^{\text{raw}} \in \mathbb{R}^{N \times d_v}$ and $\mathbf{T}^{\text{raw}} \in \mathbb{R}^{N \times d_t}$. The task is to learn representations $\mathbf{e}_u$ and $\mathbf{e}_i$ such that the scoring function $f_\Theta(\mathbf{e}_u, \mathbf{e}_i)$ accurately predicts $\hat{y}_{ui}$, where $\Theta$ denotes model parameters. In this work, we focus on eliminating cross-modal redundancy between visual and textual features by applying null-space projection with intensity $\lambda$ along the top-$k$ dominant correlated directions.


\subsection{Multimodal Recommendation Formulation}

Multimodal recommenders typically follow a general pipeline to incorporate auxiliary content features into collaborative filtering. Raw multimodal features $\mathbf{V}^{\text{raw}}$ and $\mathbf{T}^{\text{raw}}$ are first transformed into dense embeddings through feature encoders: 
\begin{equation}
\mathbf{V} = \phi_v(\mathbf{V}^{\text{raw}}), \quad \mathbf{T} = \phi_t(\mathbf{T}^{\text{raw}}),
\end{equation}
where $\phi_v$ and $\phi_t$ project raw features into a shared embedding space $\mathbb{R}^d$. These embeddings are then refined through representation learning modules, such as graph convolutional networks over user-item interactions, attention mechanisms, or contrastive learning objectives to capture collaborative signals and cross-modal dependencies. The resulting modality-specific representations are fused via weighted aggregation or concatenation to form unified item embeddings $\mathbf{e}_i$. Finally, user preferences are predicted through scoring functions, and recommendations are generated by ranking items based on these predicted scores.

Despite the effectiveness of this paradigm, a critical challenge emerges: visual and textual modalities often encode overlapping semantics, leading to \textbf{cross-modal redundancy} that dilutes modality-specific signals and causes optimization to repeatedly amplify shared components. This motivates our approach to explicitly model and suppress such redundancy, as detailed in Section~\ref{sec:methodology}.

\section{Methodology}
\label{sec:methodology}
In this section, we conduct a detailed analysis of CLEAR's key mechanisms and procedures, including the construction of cross-modal variance matrices, SVD decomposition, and null-space projection operations. Figure~\ref{fig:architecture} illustrates CLEAR's architecture and training pipeline.

\begin{figure*}[t]
    \centering
    \includegraphics[width=0.95\textwidth]{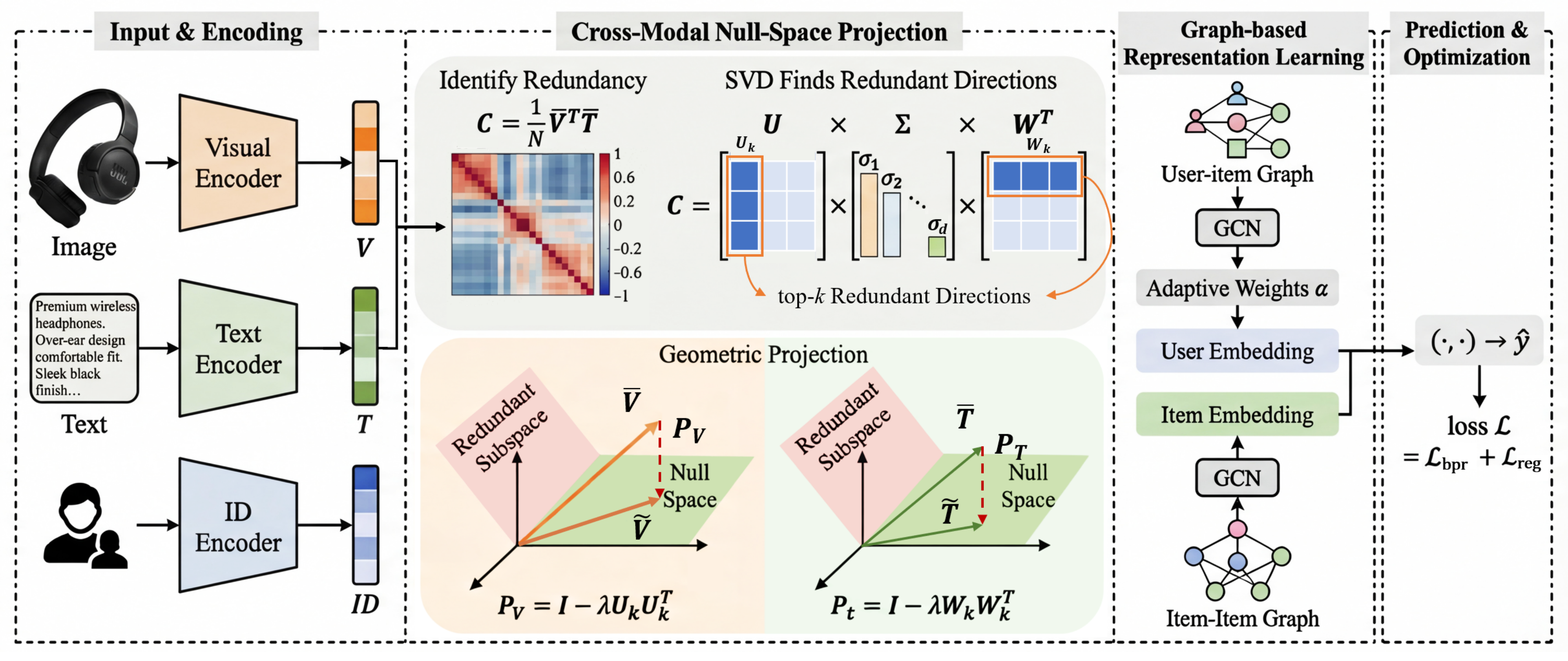}
    \caption{The overall architecture of CLEAR. Visual and textual features are encoded, then a cross-modal covariance matrix $\mathbf{C}$ is constructed and decomposed via SVD to identify top-$k$ redundant directions. Null-space projections $\mathbf{P}_V$ and $\mathbf{P}_T$ suppress redundancy while preserving modality-specific information for downstream graph-based recommendation.}
    \label{fig:architecture}
\end{figure*}

\subsection{Overview of CLEAR}

CLEAR addresses cross-modal redundancy by performing a direct geometric intervention on the learned representations rather than relying on auxiliary losses. The method begins by explicitly modeling the statistical dependencies between visual and textual features through their cross-modal covariance matrix. Then we perform SVD on this covariance matrix to identify the dominant directions in which redundant correlation energy is concentrated. Treating these leading singular vectors as the primary axes of redundancy, CLEAR applies a soft null-space projection that attenuates components along those directions while fully preserving information in the orthogonal complement so that modality-specific signals become more prominent and less overshadowed by shared semantics. The projected features are then seamlessly fed into the downstream graph neural network to capture collaborative filtering signals. In the following section, we will provide a detailed explanation of CLEAR.

\subsection{Cross-Modal Redundancy Modeling}
In our previous statements, we defined redundancy as the semantic overlap between visual and textual modalities. Specifically, this semantic overlap manifests mathematically as cross-modal covariance: when two modalities encode the same concept, their corresponding feature dimensions exhibit synchronized patterns of change across items. Following established information theory~\cite{reza1994introduction}, we define cross-modal redundancy as task-agnostic shared information. For visual features $\mathbf{V}$ and textual features $\mathbf{T}$, the total mutual information decomposes as:
\begin{equation}
    I(\mathbf{V}; \mathbf{T}) = \underbrace{I(\mathbf{V}; \mathbf{T} | Y)}_{\text{task-relevant}} + \underbrace{[I(\mathbf{V}; \mathbf{T}) - I(\mathbf{V}; \mathbf{T}|Y)]}_{\text{redundancy } \mathcal{R}}
\end{equation}
where $Y$ denotes the recommendation task labels. The redundancy term $\mathcal{R}(\mathbf{V}, \mathbf{T}|Y)$ quantifies information shared between modalities but irrelevant to recommendation performance.

\subsubsection{Characterizing Redundancy and Its Interference Mechanisms}
A crucial distinction separates redundancy from alignment is that alignment measures cross-modal correspondence, whereas redundancy captures systematic overlap. 
This clarification reveals why redundancy harms recommendation performance in three ways: \emph{representation collapse}, \emph{fusion imbalance}, and \emph{optimization pathology}. Semantic overlap primarily encodes modality-generic information, which offers limited discriminative power and struggles to capture more distinctive signals. During training, redundant information often leads to uneven information allocation, causing overemphasized weighting and attention toward duplicated coding patterns. Specifically, in SGD~\cite{ruder2016overview}, models typically prioritize fitting these repetitive patterns that are easier to learn, thus weakening the incentive to learn more discriminative modality-specific information. Without suppressing redundancy, multimodal learning gradually degenerates into amplification of unimodal signals, further exacerbating overfitting to repetitive information. We provide rigorous theoretical proofs for the above three points in Appendix~\ref{harmful} to demonstrate the necessity of addressing redundancy issues.



\subsubsection{Cross-Modal Covariance Construction}

To quantify the correlation structure between visual and textual features, we construct the cross-modal covariance matrix.
Given feature matrices $\mathbf{V} \in \mathbb{R}^{N \times d_v}$ and $\mathbf{T} \in \mathbb{R}^{N \times d_t}$ obtained after the MLP transformation, we first perform mean-centering to eliminate global offsets and focus on relational variations:
\begin{equation}
    \bar{\mathbf{V}} = \mathbf{V} - \mathbf{1}_N \boldsymbol{\mu}_v^\top, \quad \bar{\mathbf{T}} = \mathbf{T} - \mathbf{1}_N \boldsymbol{\mu}_t^\top,
\end{equation}
where $\bar{\mathbf{V}}$ and $\bar{\mathbf{T}}$ denote the mean-centered visual and textual feature matrices, respectively; $\mathbf{1}_N$ is a column vector of ones; and $\boldsymbol{\mu}_v = \frac{1}{N} \sum_{i=1}^{N} \mathbf{v}_i$ and $\boldsymbol{\mu}_t = \frac{1}{N} \sum_{i=1}^{N} \mathbf{t}_i$ are the column-wise means. Mean-centering ensures that the covariance captures true co-variation patterns rather than biases from differing magnitudes or shifts across modalities.
The cross-modal covariance matrix is then defined as:
\begin{equation}
    \mathbf{C} = \frac{1}{N} \bar{\mathbf{V}}^\top \bar{\mathbf{T}} \in \mathbb{R}^{d_v \times d_t}.
\end{equation}
This matrix $\mathbf{C}$ is derived from the empirical covariance estimator in statistics~\cite{nakada2023understanding,oh2024towards}, scaled by $1/N$ to provide an unbiased estimate of the population covariance between visual and textual features. The element $c_{ij} = \frac{1}{N} \sum_{n=1}^{N} \bar{v}_{ni} \bar{t}_{nj}$ represents the covariance between the $i$-th visual dimension and the $j$-th textual dimension across all $N$ items. Each row of $\mathbf{C}$ can be interpreted as the correlation strengths of a visual dimension (e.g., a color channel) with all textual dimensions (e.g., word embedding components). Large absolute values in $\mathbf{C}$ indicate strong co-variation, signaling potential redundancy where changes in one modality systematically predict changes in the other.

\subsubsection{Localizing Redundancy via SVD}

To extract the structure of dominant correlation, we perform Singular Value Decomposition(SVD) on the cross-modal covariance matrix:
\begin{equation}
    \mathbf{C} = \mathbf{U} \boldsymbol{\Sigma} \mathbf{W}^\top = \sum_{i=1}^{d} \sigma_i \mathbf{u}_i \mathbf{w}_i^\top
\end{equation}

where $\mathbf{U} = [\mathbf{u}_1, \ldots, \mathbf{u}_d] \in \mathbb{R}^{d \times d}$ contains the left singular vectors corresponding to the visual feature space, and $\mathbf{W} = [\mathbf{w}_1, \ldots, \mathbf{w}_d] \in \mathbb{R}^{d \times d}$ contains the right singular vectors corresponding to the textual feature space. This assignment follows directly from the construction of the cross-modal covariance matrix $\mathbf{C} = \frac{1}{N} \bar{\mathbf{V}}^\top \bar{\mathbf{T}}$, in which the rows index the visual dimensions and the columns the textual dimensions. As a result, $\mathbf{U}$ spans the visual subspace while $\mathbf{W}$ spans the textual subspace. The diagonal matrix $\boldsymbol{\Sigma} = \mathrm{diag}(\sigma_1, \ldots, \sigma_d)$ contains singular values sorted in descending order, with $\sigma_1 \geq \sigma_2 \geq \cdots \geq \sigma_d \geq 0$. Each singular value $\sigma_i$ quantifies the strength of cross-modal correlation captured by the corresponding pair of singular vectors, and its magnitude directly reflects the amount of redundant energy along that direction.

Performing SVD on the cross-modal covariance matrix essentially involves finding an orthogonal basis that maximizes the visual-textual collaborative variation across all possible linear direction pairs. The singular vector pairs are not artificially defined semantic axes but emerge automatically from the covariance structure in a globally statistical sense. The corresponding singular value $\mathbf{\sigma}_i$ quantifies the strength of cross-modal correlation along that direction. It is important to note that SVD extracts only the optimal collaborative directions in a linear sense. In real recommendation data, different semantic factors in the feature space are often not strictly orthogonal but highly entangled and overlapping. Consequently, each singular direction typically does not correspond to a single, pure semantic concept but rather resembles a linear mixture of multiple effective information components and ineffective noise. We acknowledge that decomposition operations may affect task-relevant information, yet we can still demonstrate through Theory~\ref{necessity} that redundancy suppression continues to yield performance gains.

\subsection{Null-Space Projection for De-redundancy}
Having identified redundancy, we now describe how CLEAR suppresses redundancy during training. The key insight is that redundant directions, once identified, tend to be repeatedly reinforced during training because they dominate the representations across modalities. CLEAR intervenes at the representation level by projecting the embeddings onto the subspace orthogonal to the redundant directions, thereby preventing the amplification of shared semantic components. 
Formally, given a subspace $\mathcal{S}$ spanned by the columns of a matrix $\mathbf{M} \in \mathbb{R}^{d \times k}$ with orthonormal columns, a vector $\mathbf{v} \in \mathbb{R}^d$ is said to lie in the null space of $\mathcal{S}$ if and only if
\begin{equation}
\mathbf{M}^\top \mathbf{v} = \mathbf{0}.
\end{equation}
The corresponding orthogonal projection operator onto this null space is defined as
\begin{equation}
\mathbf{P}^{\perp} = \mathbf{I} - \mathbf{M}\mathbf{M}^\top,
\end{equation}
where $\mathbf{I}$ denotes the identity matrix. Geometrically, applying $\mathbf{P}^{\perp}$ to a vector removes its components along the directions spanned by $\mathbf{M}$, retaining only the components orthogonal to this subspace. Theorem~\ref{thm:suppression} demonstrates the effectiveness of the null-space projection in eliminating redundancy.

Let $\mathbf{U}_k = [\mathbf{u}_1, \ldots, \mathbf{u}_k] \in \mathbb{R}^{d \times k}$ and $\mathbf{W}_k = [\mathbf{w}_1, \ldots, \mathbf{w}_k] \in \mathbb{R}^{d \times k}$ denote the top-$k$ left and right singular vectors. The orthogonal projector onto the subspace spanned by $\mathbf{U}_k$ is $\mathbf{U}_k \mathbf{U}_k^\top$. To project onto its null-space, we use:
\begin{equation}
    \mathbf{P}_v^{\perp} = \mathbf{I} - \mathbf{U}_k \mathbf{U}_k^\top, \quad \mathbf{P}_t^{\perp} = \mathbf{I} - \mathbf{W}_k \mathbf{W}_k^\top
\end{equation}
where $\mathbf{I} \in \mathbb{R}^{d \times d}$ is the identity matrix.

To enable flexible control over suppression strength and maintain training stability, we introduce soft projection with a parameter $\lambda \in [0,1]$:
\begin{equation}
    \mathbf{P}_v = \mathbf{I} - \lambda \mathbf{U}_k \mathbf{U}_k^\top, \quad \mathbf{P}_t = \mathbf{I} - \lambda \mathbf{W}_k \mathbf{W}_k^\top
\end{equation}
When $\lambda = 0$, features pass through unchanged; when $\lambda = 1$, full null-space projection is applied. Intermediate values allow partial suppression, achieving a balance between redundancy reduction and information preservation. The projected features are:
\begin{equation}
    \tilde{\mathbf{V}} = \bar{\mathbf{V}} \mathbf{P}_v, \quad \tilde{\mathbf{T}} = \bar{\mathbf{T}} \mathbf{P}_t
\end{equation}

Geometrically, this operation attenuates components along the top-$k$ redundant directions by a factor of $(1 - \lambda)$, while leaving orthogonal components intact. Although SVD cannot achieve perfect separation between essential discriminative signals and redundant semantic overlap, choosing $\lambda < 1$ allows the model to preserve partial information even along redundant directions, thereby mitigating collateral damage to task-relevant signals that are entangled with redundancy.

\subsection{Theoretical Foundation: Low-Rank Redundancy Hypothesis}

Our method is grounded in a central theoretical hypothesis: \emph{cross-modal redundancy concentrates in a low-dimensional subspace whose intrinsic dimensionality is fixed and data-dependent}. This perspective fundamentally distinguishes our approach from adaptive strategies and provides strong justification for employing a fixed redundancy rank $k$ throughout training.

Following the manifold hypothesis~\cite{fefferman2016testing,brown2022verifying}, high-dimensional multimodal data lies on or near low-dimensional manifolds. In recommendation scenarios, visual and textual modalities share a common semantic manifold $\mathcal{M}_s$ with intrinsic dimension $d_s \ll \min(D_v, D_t)$, formally defined as the smallest $m$ such that the data can be faithfully approximated by an $m$-dimensional manifold:
\begin{equation}
d_s = \min\left\{m : \frac{1}{n}\sum_{i=1}^n \|\mathbf{x}_i - \text{proj}_{\mathcal{M}_m}(\mathbf{x}_i)\|^2 < \epsilon\right\}.
\end{equation}
This shared manifold primarily arises from a limited set of modality-agnostic semantic factors—most notably category-level concepts, whose number $c$ is typically far smaller than both the feature dimensionality $d$ and the sample size $N$. Consequently, the cross-modal covariance matrix $\mathbf{C}$ exhibits a pronounced low-rank structure~\cite{zheng2024learning}:
\begin{equation}
\text{rank}(\mathbf{C}) \leq c \ll \min(N, d).
\end{equation}

Critically, recent studies on neural network optimization~\cite{li2018measuring} demonstrate that such intrinsic dimensionality is determined by the underlying data distribution rather than transient training dynamics; the effective degrees of freedom of the objective landscape remain stable throughout learning. Applied to cross-modal redundancy, this implies that the optimal number of redundant directions $k^*$ is essentially a dataset constant rather than a quantity that should evolve during optimization:
\begin{equation}
k^* = \argmin_{k} \left\{ \sum_{i > k} \sigma_i^2(\mathbf{C}) + \lambda k \right\}.
\end{equation}
Therefore, we adopt a \textbf{fixed-rank projection} strategy in which $k$ is selected once on the validation set and held constant during training. This design ensures a stable projection subspace, allowing the model to consistently suppress the same redundant directions and reliably preserve modality-specific information in the orthogonal complement. In contrast, adaptive rank selection (e.g., energy-ratio thresholding) introduces undesirable instability: singular values fluctuate across mini-batches and optimization steps, resulting in a constantly shifting projection target. Moreover, high-energy components may stem from modality-specific noise rather than true semantic redundancy~\cite{rehman2025multimodal}, making purely spectral adaptive methods prone to error.

\subsection{Optimization and Training}

Our training procedure integrates null-space projection seamlessly into the standard recommendation optimization pipeline without introducing auxiliary loss functions. The projection matrices are computed at the epoch level and cached to balance computational efficiency with feature distribution adaptability, while model parameters are updated via mini-batch gradient descent. Specifically, at the beginning of each training epoch (or every $\tau$ epochs), we update the projection matrices $\mathbf{P}_v$ and $\mathbf{P}_t$ based on the current MLP-transformed features, which are then applied consistently across all mini-batches within that epoch. During forward propagation, the projected features $\tilde{\mathbf{V}} = \mathbf{V}\mathbf{P}_v$ and $\tilde{\mathbf{T}} = \mathbf{T}\mathbf{P}_t$ flow through graph convolutional layers to generate modality-specific embeddings, which are subsequently fused via learnable modal-specific attention weights and further refined through item-item graph aggregation to construct the final user and item representations for preference prediction. Following standard practice in implicit feedback recommendation, we adopt the Bayesian Personalized Ranking (BPR) loss~\cite{rendle2012bpr} to optimize pairwise preference rankings:
\begin{equation}
\mathcal{L}_{\text{BPR}} = - \sum_{(u,i,j) \in \mathcal{O}} \log \sigma(\hat{y}_{ui} - \hat{y}_{uj})
\end{equation}
where $\mathcal{O} = \{(u, i, j) : r_{ui} = 1, r_{uj} = 0\}$ is the set of training triplets with observed positive interactions and sampled negative items, $\hat{y}_{ui} = \langle \mathbf{e}_u, \mathbf{e}_i \rangle$ denotes the predicted preference score, and $\sigma(\cdot)$ is the sigmoid function. To prevent overfitting, we apply $L_2$ regularization on the user preference embeddings and fusion weights:
\begin{equation}
\mathcal{L}_{\text{reg}} = \|\mathbf{E}_u^v\|_2^2 + \|\mathbf{E}_u^t\|_2^2 + \|\mathbf{w}\|_2^2
\end{equation}
where $\mathbf{E}_u^v$ and $\mathbf{E}_u^t$ denote the visual and textual user preference embeddings before fusion, and $\mathbf{w}$ represents the learnable modal fusion weights. The final training objective combines the BPR loss with regularization:
\begin{equation}
\mathcal{L} = \mathcal{L}_{\text{BPR}} + \gamma \mathcal{L}_{\text{reg}}
\end{equation}
where $\gamma$ is the regularization coefficient. Crucially, the SVD computation is performed on detached features to prevent gradients from flowing through the projection matrix construction, ensuring that $\mathbf{P}_v$ and $\mathbf{P}_t$ serve as fixed geometric constraints rather than learnable parameters within each epoch. This design allows the model to learn representations that are both discriminative for recommendation and compatible with redundancy suppression.

\begin{table}[t]
\centering
\caption{Statistics of the experimental datasets.}
\label{tab:datasets}
\begin{tabular}{lcccc}
\toprule
\textbf{Dataset} & \textbf{\#Users} & \textbf{\#Items} & \textbf{\#Inters} & \textbf{Sparsity} \\
\midrule
Baby     & 19,445 & 7,050  & 160,792 & 99.88\% \\
Sports   & 35,598 & 18,357 & 296,337 & 99.95\% \\
Clothing & 39,387 & 23,033 & 278,677 & 99.97\% \\
\bottomrule
\end{tabular}
\end{table}

\section{Experiments}

\begin{table*}[!ht]
    \centering
    \caption{Performance comparison of baselines and our framework in terms of Recall@K (R@\textit{K}), and NDCG@K (N@\textit{K}).}
    \tabcolsep=0.1cm
    \vskip -0.1in
    \label{tab:overall}
    \resizebox{0.92 \textwidth}{!}{ 
    \begin{tabular}{r|cccc|cccc|cccc}
    \toprule[1.2pt]
         \textbf{Datasets}&  \multicolumn{4}{c}{\textbf{Baby}}&  \multicolumn{4}{|c}{\textbf{Sports}}&  \multicolumn{4}{|c}{\textbf{Clothing}}\\
         \midrule
         \midrule
         \textbf{Methods} & \textbf{R@10} & \textbf{R@20} & \textbf{N@10} & \textbf{N@20} & \textbf{R@10} & \textbf{R@20} & \textbf{N@10} & \textbf{N@20} & \textbf{R@10} & \textbf{R@20} & \textbf{N@10} & \textbf{N@20} \\ 
         \midrule
         MF-BPR & 0.0357& 0.0575& 0.0192& 0.0249& 0.0432& 0.0653& 0.0241& 0.0298 & 0.0187 & 0.0279 & 0.0103 & 0.0126\\
         LightGCN & 0.0479& 0.0754& 0.0257& 0.0328& 0.0569& 0.0864& 0.0311& 0.0387 & 0.0340 & 0.0526 & 0.0188 & 0.0236\\
         LayerGCN & 0.0529& 0.0820& 0.0281& 0.0355& 0.0594& 0.0916& 0.0323& 0.0406& 0.0371& 0.0566& 0.0200& 0.0247\\
         \midrule
         VBPR &  0.0423&  0.0663&  0.0223&  0.0284&  0.0558&   0.0856& 0.0307& 0.0384& 0.0281& 0.0415& 0.0158& 0.0192\\
         MMGCN &  0.0378&  0.0615&  0.0200&  0.0261&  0.0370&   0.0605& 0.0193& 0.0254 & 0.0218& 0.0345& 0.0110& 0.0142\\
         DualGNN &  0.0448&  0.0716&  0.0240&  0.0309&  0.0568&   0.0859& 0.0310& 0.0385& 0.0454& 0.0683& 0.0241& 0.0299\\
         LATTICE &  0.0547&  0.0850&  0.0292&  0.0370&  0.0620&  0.0953&  0.0335&  0.0421& 0.0492& 0.0733& 0.0268& 0.0330\\
         SLMRec & 0.0529& 0.0775& 0.0290& 0.0353& 0.0663&   0.0990& 0.0365& 0.0450 &0.0452& 0.0675& 0.0247& 0.0303\\
         BM3 & 0.0564& 0.0883& 0.0301& 0.0383& 0.0656& 0.0980& 0.0355& 0.0438& 0.0422& 0.0621& 0.0231& 0.0281\\
         MMSSL & 0.0613& 0.0971& 0.0326& 0.0420& 0.0673& 0.1013& 0.0380& 0.0474& 0.0531& 0.0797& 0.0291& 0.0359\\
         FREEDOM & 0.0627& 0.0992& 0.0330& 0.0424& 0.0717& 0.1089& 0.0385& 0.0481& 0.0629& 0.0941& 0.0341& 0.0420\\
         MGCN & 0.0620& 0.0964& 0.0339& 0.0427& 0.0729 & 0.1106 & 0.0397 & 0.0496 & 0.0641 & 0.0945 & 0.0347 & 0.0428 \\
         LGMRec & 0.0644 & 0.1002& \underline{0.0349}& 0.0440 &0.0720 & 0.1068& 0.0390 & 0.0480  & 0.0555 & 0.0828 & 0.0302& 0.0371 \\
          MENTOR & \underline{0.0648}& \underline{0.1035}& \underline{0.0349}& \underline{0.0448}& \underline{0.0736}& \underline{0.1130}& \underline{0.0399}& \underline{0.0501} & \underline{0.0659}& \underline{0.0976}& \underline{0.0357}& \underline{0.0437}\\
         \rowcolor{black!15} CLEAR &  \textbf{0.0687} &  \textbf{0.1045}& \textbf{0.0364} &  \textbf{0.0456} &  \textbf{0.0756}&  \textbf{0.1146} &\textbf{0.0411} & \textbf{0.0511}
         &  \textbf{0.0675}&  \textbf{0.1001} &\textbf{0.0358} & \textbf{0.0440}  \\
         \bottomrule[1.2pt]
    \end{tabular}}
\end{table*}

\subsection{Experimental Settings}
In this section, we conduct extensive experiments to verify our assumptions and demonstrate the effectiveness, generalization, and efficiency. Our experimental results provide solid answers to the following key questions.

\begin{itemize}
    \item \textbf{RQ1:} Does our CLEAR architecture consistently outperform state-of-the-art methods?
    \item \textbf{RQ2:} Do our design principles hold in practice?
    \item \textbf{RQ3:} What is the impact of hyperparameters on performance?
    \item \textbf{RQ4:} What are the computational costs, general applicability, and representation improvements of CLEAR?
\end{itemize}

\subsubsection{Datasets.}
We conduct empirical studies on three widely used datasets from Amazon review data~\cite{mcauley2015image}: \textbf{Baby}, \textbf{Sports}, and \textbf{Clothing}, which include adequate visual and information.Table~\ref{tab:datasets} summarizes the statistics of these datasets. Following previous works~\cite{wei2019mmgcn,he2016vbpr}, we apply 5-core filtering to ensure data quality. Furthermore, the pre-extracted 4096-dimensional visual features are utilized, and 384-dimensional textual features are extracted using a pre-trained sentence transformer. Appendix ~\ref{app:dataset} provides details on data collection.

\subsubsection{Baselines}
We compared our CLEAR model with fourteen classic recommendation models, including three traditional collaborative filtering-based models: MF-BPR~\cite{rendle2012bpr}, LightGCN~\cite{he2020lightgcn}, and LayerGCN~\cite{zhou2023layer}; and eleven multimodal recommendation models: VBPR~\cite{he2016vbpr}, MMGCN~\cite{wei2019mmgcn}, DualGNN~\cite{wang2021dualgnn}, LATTICE~\cite{zhang2021LATTICE}, SLMRec~\cite{tao2022self}, BM3~\cite{zhou2023bootstrap}, MMSSL~\cite{wei2023multi}, FREEDOM~\cite{zhou2023tale}, MGCN~\cite{yu2023multi}, LGMRec~\cite{guo2024lgmrec}, and MENTOR~\cite{xu2025mentor}. The full list of baselines and their implementation details are outlined in Appendix~\ref{app:baselines}.

\subsubsection{Evaluation Protocals:}
To ensure a fair evaluation of our approach, we adopt Recall@K and NDCG@K with $K \in \{10, 20\}$. We randomly split the data into training, validation and test sets in an 8:1:1 distribution.

\subsubsection{Implementation Details.}
We leverage the MMRec~\cite{zhou2023mmrec} open-source framework for model development and reproduction of various baselines. Our experiments are conducted using PyTorch and conduct all experiments on a single NVIDIA RTX 4090 GPU with 24GB memory. 
For our model, we set the embedding dimension to 64 for both user and item embeddings. All training parameters were initialized using the Xavier~\cite{glorot2010understanding}, and the Adam optimizer~\cite{kinga2015method} was adopted during training.
For the null-space projection module, we search the projection strength $\lambda$ in $\{0.3, 0.5, 0.7, 0.9\}$ and the redundancy rank $k$ in $\{2, 4, 8, 16, 20\}$. The optimal hyperparameters are selected based on performance on the validation set.

\subsection{Overall Performance(RQ1)}

Table~\ref{tab:overall} presents the performance comparison of our method against all baselines on three datasets. We have the following observations:





\textbf{Methods considering both modality fusion and disentanglement outperform those focusing solely on alignment}
Contrastive learning methods like SLMRec, BM3, and MMSSL learn multimodal representations by pulling cross-modal features closer, achieving competitive performance over basic multimodal architectures. However, methods that explicitly balance modality alignment with disentanglement consistently demonstrate superior performance in comparison to alignment-only approaches. For instance, MGCN employs a dual-channel architecture that explicitly separates shared and unique representations, thereby outperforming alignment-based methods across various datasets.

\textbf{Our method performs better than both alignment-based and disentanglement-based approaches.}
Our approach simultaneously incorporates two core principles while transcending the traditional trade-off framework: it neither enforces alignment nor presumes shared-specific boundaries. Instead, it directly identifies and suppresses redundant directions through data-driven covariance analysis. Across various datasets, using the R@20 metric as an example, CLEAR achieves an improvement ranging from 7.6\% to 25.6\% over MMSSL and from 3.6\% to 5.9\% over MGCN. The findings of this study demonstrate that the precise localisation and suppression of redundant directions is a more efficacious approach than either reinforcing alignment or relying on predefined boundaries.

\textbf{Our method achieves state-of-the-art performance across all datasets with consistent improvements.}
Compared to the strongest baseline on each dataset, our method achieves improvements of 6.0\% (Baby), 2.7\% (Sports), and 5.3\% (Clothing) on Recall@10, with NDCG metrics showing similar trends. The three datasets exhibit diverse characteristics in scale and sparsity. The consistent gains across these diverse scenarios indicate that cross-modal redundancy is a pervasive phenomenon that existing methods have not adequately addressed, and our redundancy suppression strategy generalizes well regardless of dataset properties.

\subsection{Ablation Study(RQ2)}
\label{sec:ablation}

To validate our key design choices, we conduct ablation experiments examining three configurations: (1) full model with null-space projection using fixed top-$k$, (2) removing null-space projection entirely, and (3) replacing fixed $k$ with dynamic ratio-based selection that determines $k$ based on singular value energy distribution. Figure~\ref{fig:ablation} summarizes the results.

\begin{figure}[t]
    \centering
    \includegraphics[width=\linewidth]{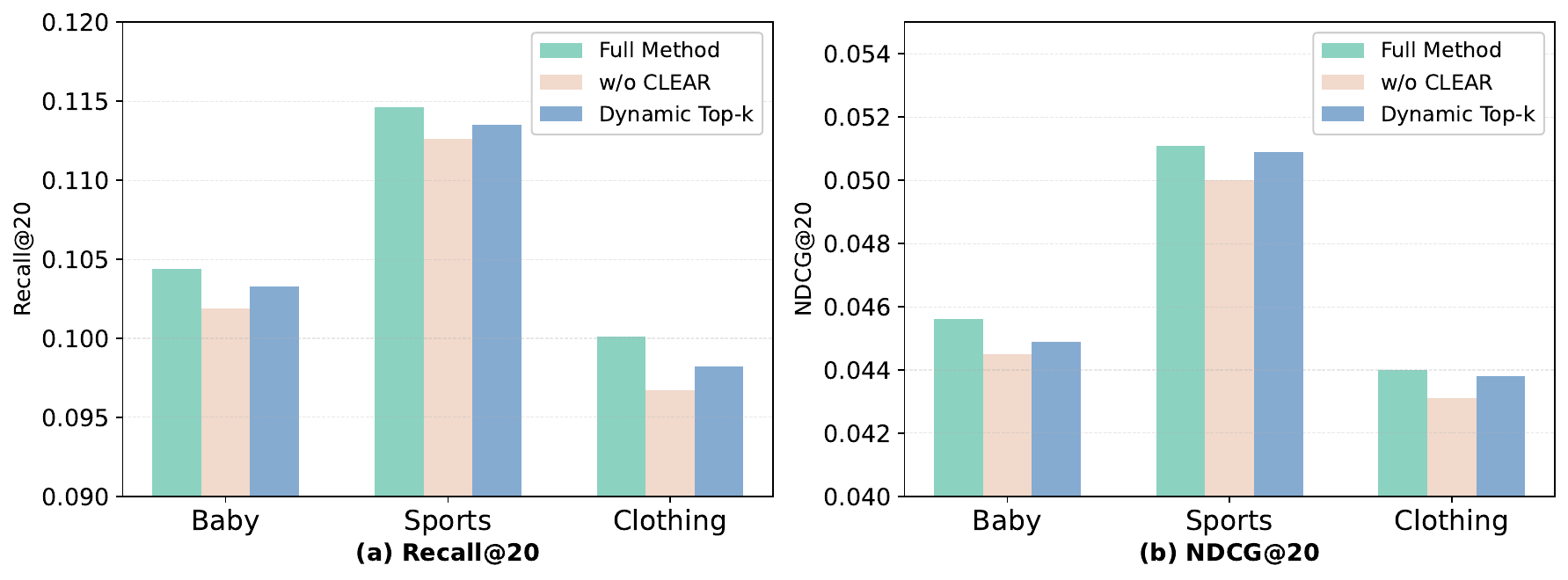}
    \vspace{-0.4cm}
    \caption{Ablation study comparing variants on three datasets. The bar chart shows R@20 and N@20 metrics for: Full Model (with null-space projection and fixed top-$k$), w/o null-space, and Dynamic Ratio (adaptive ratio-based selection).}
    \label{fig:ablation}
    \vspace{-0.2cm}
\end{figure}

\textbf{The Role of Null-Space Projection.}
As shown in Figure~\ref{fig:ablation}, removing null-space projection (w/o null-space) leads to consistent performance drops across all datasets. On Recall@20, the degradation reaches 2.4\% on Baby, 1.7\% on Sports, and 3.4\% on Clothing; NDCG@20 exhibits similar trends with drops of 2.0\%--2.4\%. These results validate our core hypothesis: cross-modal redundancy impedes recommendation performance by allowing modal-generic factors to overshadow modality-specific discriminative information. The null-space projection effectively suppresses these redundant directions, enabling the model to better leverage the unique contributions of each modality.

\textbf{Fixed Top-$k$ vs. Dynamic Ratio Selection.}
The full method with fixed $k$ consistently outperforms dynamic ratio selection. The underlying reason is that dynamic selection treats the redundancy rank as a training-dependent quantity, creating a moving optimization target as the singular value spectrum evolves. This leads to instability where the model continuously adapts to shifting projection subspaces. In contrast, fixed $k$ provides a constant regularization target throughout training, allowing the model to learn stable representations. The results support that cross-modal redundancy concentrates in a low-dimensional subspace whose dimensionality is an intrinsic property of the data domain rather than a transient characteristic of training dynamics.

\begin{table}[t]
\centering
\caption{Generalization analysis on Baby dataset.}
\label{tab:generalization}
\begin{tabular}{lcccc}
\toprule
\textbf{Model} & \textbf{R@10} & \textbf{R@20} & \textbf{N@10} & \textbf{N@20} \\
\midrule
VBPR & 0.0423 & 0.0663 & 0.0223 & 0.0284 \\
VBPR + CLEAR & 0.0469 & 0.0776 & 0.0258 & 0.0336 \\
\rowcolor{gray!15} \textit{Improv.} & \textit{10.87\%} & \textit{17.04\%} & \textit{15.70\%} & \textit{18.31\%} \\
\midrule
FREEDOM & 0.0624 & 0.0984 & 0.0328 & 0.0420 \\
FREEDOM + CLEAR & 0.0634 & 0.0992 & 0.0337 & 0.0427 \\
\rowcolor{gray!15} \textit{Improv.} & \textit{1.60\%} & \textit{0.81\%} & \textit{2.74\%} & \textit{1.67\%} \\
\bottomrule
\end{tabular}
\end{table}

\subsection{Hyperparameter Analysis(RQ3)}

\subsubsection{Effect of Projection Strength $\lambda$.}
Both datasets show consistent improvement as $\lambda$ increases, with optimal performance at $\lambda = 0.9$. This validates that stronger redundancy suppression generally benefits recommendation by allowing modality-specific semantics to dominate. The stability at high $\lambda$ values indicates that our method effectively suppresses modal-generic factors without removing task-relevant information.
\subsubsection{Effect of Redundancy Rank $k$.}
The heatmaps reveal an inverted-U relationship: very small $k$ under-suppresses redundancy, while excessively large $k$ risks removing useful information. Optimal values differ by dataset---Clothing favors $k=4$ while Sports prefers $k=8$---reflecting domain-specific redundancy structures. Notably, performance remains stable across moderate $k$ ranges, demonstrating robustness to hyperparameter selection.

\begin{figure}[t]
    \centering
    \includegraphics[width=\linewidth]{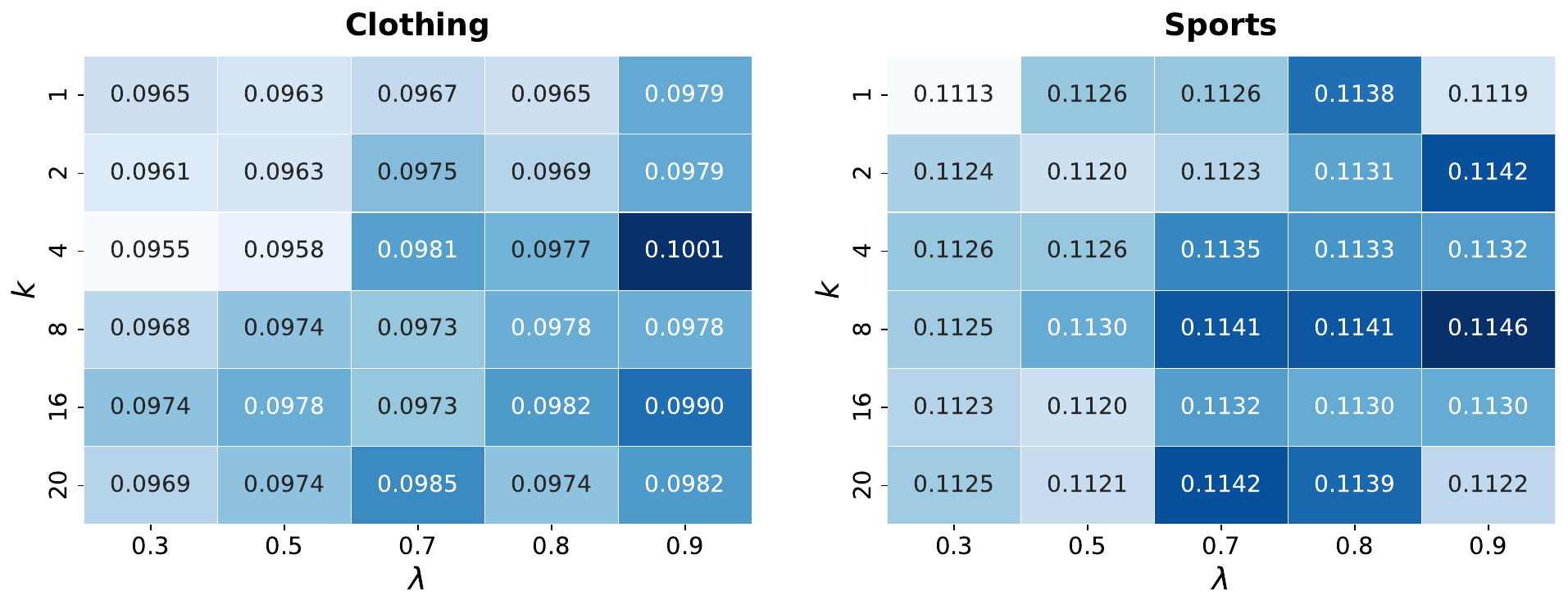}
    \vspace{-0.45cm}
    \caption{Hyperparameter analysis on Clothing and Sports datasets. Heatmaps show R@20 performance across different $\lambda$ and $k$ combinations. Darker colors indicate better results.}
    \label{fig:hyperparameter}
    \vspace{-0.45cm}
\end{figure}

\subsection{Framework efficiency, generalization and visualization Analysis(RQ4)}

\subsubsection{Computational Efficiency}

We analyze the computational overhead introduced by our null-space projection module. The time complexity is dominated by SVD decomposition $O(d^3)$ and feature projection $O(Nd^2)$, which is negligible since $d \ll N$ and the projection matrix is cached. The spatial complexity adds only $O(d^2)$ for storing two projection matrices. Algorithm ~\ref{alg:training} summarizes the training workflow of our method. Detailed experimental results on computational complexity are presented in the Table ~\ref{effciency}. The results show that our model is competitive in terms of both training efficiency and memory usage.


\subsubsection{Generalization analysis}
To verify that our null-space projection is a general framework applicable to various backbones, we integrate it into two representative multimodal recommendation models with distinct architectures: Table~\ref{tab:generalization} presents the results on the Baby dataset. We observe that both VBPR and FREEDOM have different performance improvements after null-space projection, and FREEDOM’s performance improvement is slightly lower due to its own powerful multi-graph convolutional structure.

\begin{figure}[t]
\centering
\includegraphics[width=\linewidth]{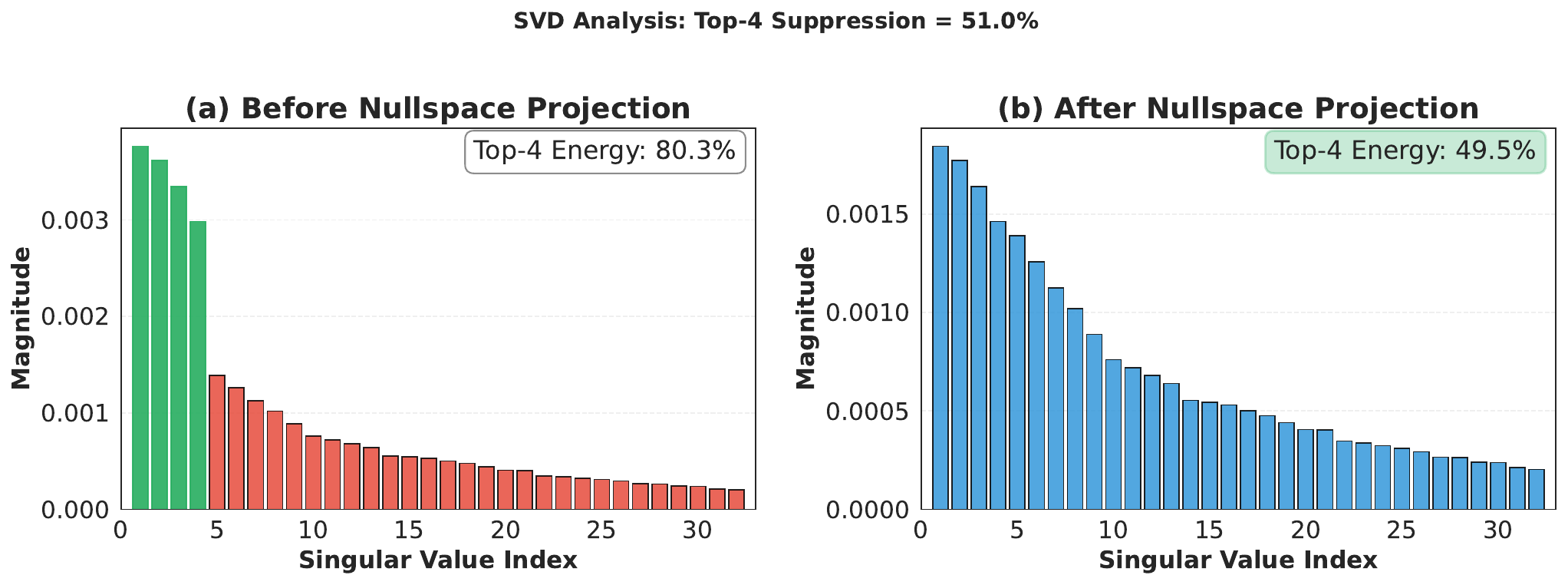}
\caption{Singular value distribution of cross-modal covariance matrix before and after null-space projection.}
\label{fig:svd_compare}
\vspace{-0.2cm}
\end{figure}

\subsubsection{Visualization Analysis}
To intuitively understand how null-space projection affects the learned representations, we provide a visualization analysis. \textbf{Singular Value Distribution.} Figure~\ref{fig:svd_compare} reveals a long-tailed singular value distribution on the Baby dataset; after null-space projection, the dominant singular values are substantially suppressed, and the energy distribution becomes markedly flatter. \textbf{t-SNE~\cite{maaten2008visualizing} Visualization.}
Figure~\ref{fig:tsne} illustrates the distribution of projects before and after null-space projection. The top and right display the marginal distributions. Through overlap and SWD metrics, we can visually observe that our CLEAR method significantly reduces redundancy at the representation level. Theorem~\ref{theorem:swd} proves this assertion.

\section{Related Works}

\subsection{Model Architectures in MMRec}
Model-level methods focus on designing sophisticated architectures to capture high-order interactions between users, items, and multimodal features. 
Early endeavors primarily utilized Graph Convolutional Networks (GCNs) to propagate multimodal signals. For example, MMGCN~\cite{wei2019mmgcn} constructs modality-specific bipartite graphs, while LATTICE~\cite{zhang2021LATTICE} learns latent item-item structures to augment original interaction graph. 
However, recent studies have begun to question the necessity of complex graph structures. FREEDOM~\cite{zhou2023tale} observes that raw multimodal graphs often contain noise and proposes a degree-guided denoising strategy to improve robustness. 
Moving beyond GCNs, TMLP~\cite{huang2025beyond} abandons expensive graph convolutions in favor of a lightweight MLP-based architecture, demonstrating that efficient feature transformation can outperform complex message passing if the multimodal signals are properly utilized. 
Despite their architectural innovations, these model-level methods often treat multimodal features as static inputs, failing to address the inherent semantic redundancy during the training process.

\subsection{Optimization Strtegies in MMRec}
Recently, a significant shift has occurred toward optimization-level innovations, focusing on how objective functions and training paradigms shape the representation space. 
Contrastive learning has become the dominant paradigm. Methods like SLMRec~\cite{tao2022self} and BM3~\cite{zhou2023bootstrap} employ self-supervised objectives to align modality-specific representations with collaborative signals. 
Furthermore, the field is evolving towards ID-free or ID-agnostic recommendation to mitigate the cold-start problem and ID-bias~\cite{yuan2023go}. Some works~\cite{li2025id} advocate for removing ID embeddings entirely or using softmax-based alignment losses instead of traditional BPR loss to achieve better cross-modal convergence. 
To handle the phenomenon of modality collapse~\cite{chaudhuri2025closer}, disentanglement techniques have been introduced. MGCN~\cite{yu2023multi} and Dual-Alignment methods attempt to separate modality-shared preferences from modality-specific ones. 
However, these optimization strategies typically rely on manually designed auxiliary losses or complex orthogonality constraints, which often lead to a trade-off between alignment and discriminability. Our work differs by introducing a geometric perspective to eliminate redundancy via null-space projection without sacrificing the integrity of the primary task.

\begin{figure}[t]
\centering
\includegraphics[width=\linewidth]{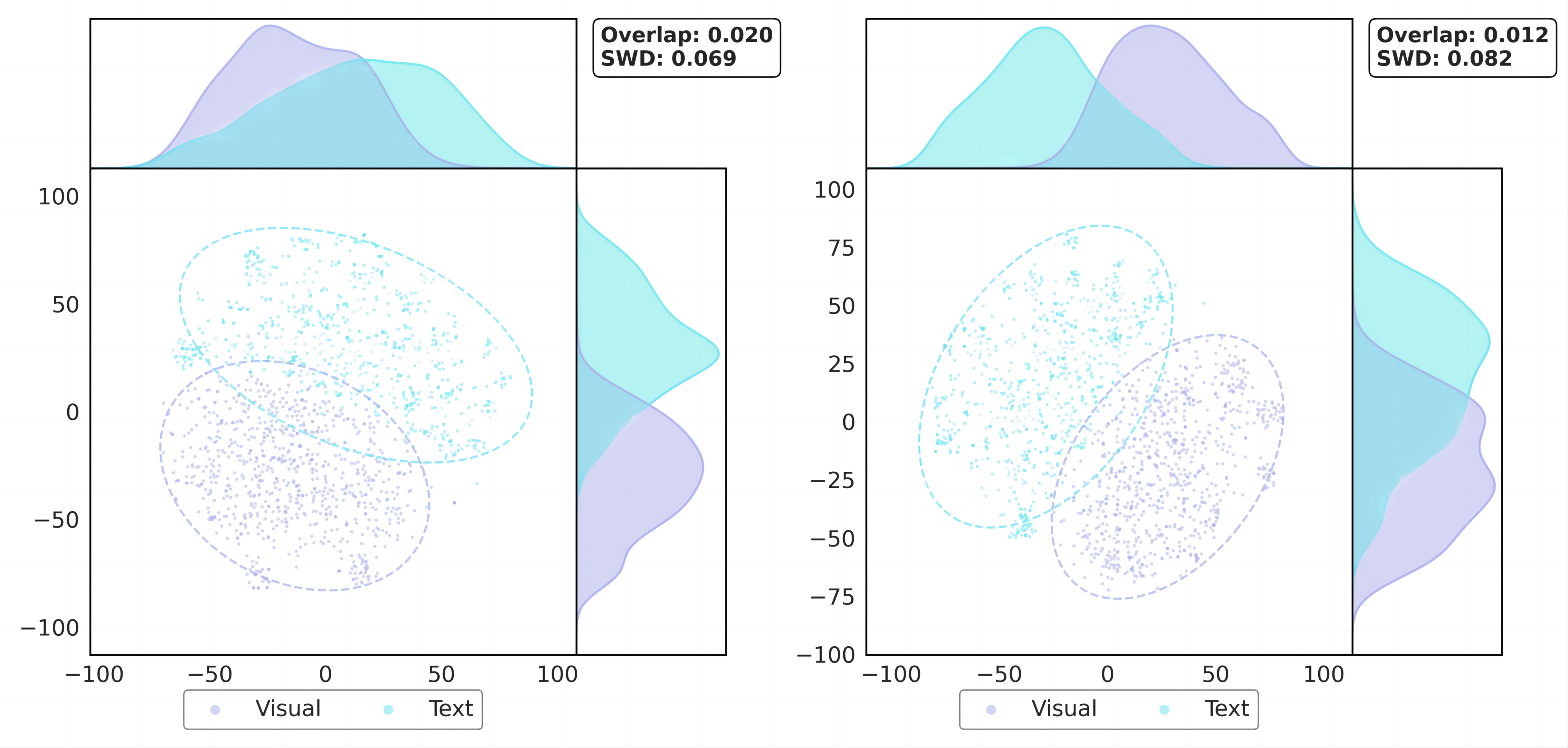}
\caption{t-SNE visualization on the Baby dataset. Our Null-space Projection reduces Overlap and increases Sliced Wasserstein Distance(SWD) significantly.}
\label{fig:tsne}
\vspace{-0.5cm}
\end{figure}

\subsection{Null-Space Projection in Machine Learning}
Null-space projection, which projects vectors onto the orthogonal complement of a given subspace, has demonstrated effectiveness across diverse machine learning applications, including fairness learning and representation learning. For instance, in the recommendation domain, 
AlphaFuse~\cite{hu2025alphafuse} decomposes language embedding space into a semantically-rich row space and a semantically-sparse null-space via SVD, injecting collaborative signals into the null-space to preserve semantic information while incorporating behavioral patterns. For fairness in classification, INLP~\cite{ravfogel2020null} iteratively trains linear classifiers to predict sensitive attributes and projects representations onto the null-space of classifier weights, effectively removing discriminatory correlations while preserving task-relevant information. In the context of large language models, AlphaEdit~\cite{fang2024alphaedit} projects parameter perturbations onto the null-space of preserved knowledge representations during knowledge editing, ensuring that updates do not alter outputs for existing facts.
In this paper, we introduce null-space projection to multimodal recommendation for cross-modal redundancy elimination. 


\section{Conclusion}

In this paper, we identify and formalize the cross-modal redundancy problem in multimodal recommendation. To address this issue, we propose CLEAR, which suppresses redundancy through null-space projection, encouraging visual and textual modalities to learn complementary rather than repetitive representations. Our approach requires no predefinition of redundant or specific information boundaries, achieves redundancy suppression through pure geometric projection without auxiliary losses, and serves as a plug-and-play module applicable to any multimodal framework. Experiments on three datasets validate consistent performance improvements. Our approach can be applied to tasks involving more than two modalities. Beyond multimodal recommendation, our method can also be used for cross-modal retrieval, where cross-modal redundancy similarly impedes complementary learning. Future work may involve more precise cross-modal information construction and decomposition.

\bibliographystyle{ACM-Reference-Format}
\balance
\bibliography{Nullspace}

\clearpage

\appendix

\section{Proof of THEOREMS}
\subsection{Sliced Wasserstein Distance (SWD)}
\label{theorem:swd}
SWD measures the overall divergence between distributions. Given two sets of $d$-dimensional features $\mathbf{V} = \{\mathbf{v}_i\}_{i=1}^N$ and $\mathbf{T} = \{\mathbf{t}_i\}_{i=1}^N$, SWD approximates the Wasserstein distance through random one-dimensional projections:
\begin{equation}
    \text{SWD}(\mathbf{V}, \mathbf{T}) = \frac{1}{L}\sum_{l=1}^{L} W_1\left(\boldsymbol{\theta}_l^\top \mathbf{V}, \boldsymbol{\theta}_l^\top \mathbf{T}\right), \quad \boldsymbol{\theta}_l \in \mathbb{S}^{d-1}
    \label{eq:swd}
\end{equation}
where $\{\boldsymbol{\theta}_l\}_{l=1}^{L}$ are random unit vectors sampled uniformly from the $(d-1)$-dimensional unit sphere, and $W_1(\cdot, \cdot)$ denotes the 1-Wasserstein distance between one-dimensional projected distributions, computed via sorted sample quantiles. Higher SWD values indicate greater distributional divergence, implying that each modality preserves more distinctive information.

\subsection{Redundancy Harms Recommendation: A Three-Fold Analysis}
\label{harmful}
\paragraph{(i) Representation Collapse.}

\begin{theorem}[Discriminability Degradation]
\label{thm:discriminability}
Decompose features as $\mathbf{v} = \mathbf{v}_{\parallel} + \mathbf{v}_{\perp}$ where $\mathbf{v}_{\parallel} = \mathbf{U}_k\mathbf{U}_k^\top\mathbf{v}$ lies in redundant subspace. For items from different classes $c_i \neq c_j$, the inter-class separation measured by Fisher's Linear Discriminant satisfies:
\begin{equation}
    \text{FDR}_{\parallel} := \frac{\|\mathbb{E}_{i \in c_1}[\mathbf{v}_{i,\parallel}] - \mathbb{E}_{j \in c_2}[\mathbf{v}_{j,\parallel}]\|^2}{\text{Var}_{i \in c_1}[\mathbf{v}_{i,\parallel}] + \text{Var}_{j \in c_2}[\mathbf{v}_{j,\parallel}]} \ll \text{FDR}_{\perp}
\end{equation}
where $\text{FDR}_{\perp}$ is computed analogously for $\mathbf{v}_{\perp}$.
\end{theorem}

\begin{proof}
The redundant subspace captures category-level semantics shared across modalities. Within-category variance dominates between-category variance in $\text{span}(\mathbf{U}_k)$:
\begin{equation}
    \text{Var}_{\text{within}}[\mathbf{v}_{\parallel}] = \mathbb{E}_c[\text{Var}_{i \in c}[\mathbf{U}_k^\top\mathbf{v}_i]] \approx \sigma_1^2
\end{equation}
\begin{equation}
    \text{Var}_{\text{between}}[\mathbf{v}_{\parallel}] = \text{Var}_c[\mathbb{E}_{i \in c}[\mathbf{U}_k^\top\mathbf{v}_i]] = O(\sigma_k^2) \ll \sigma_1^2
\end{equation}
Thus $\text{FDR}_{\parallel} = O(\sigma_k^2/\sigma_1^2) \to 0$ as spectral gap increases.
\end{proof}

\paragraph{(ii) Fusion Imbalance.}

\begin{proposition}[Redundancy Amplification in Fusion]
For weighted fusion $\mathbf{h} = \alpha\mathbf{v} + \beta\mathbf{t}$ with $\alpha, \beta > 0$, decompose $\mathbf{v} = \mathbf{v}_{\parallel} + \mathbf{v}_{\perp}$ and $\mathbf{t} = \mathbf{t}_{\parallel} + \mathbf{t}_{\perp}$. The expected energy distribution is:
\begin{equation}
    \mathbb{E}[\|\mathbf{h}\|^2] = \underbrace{(\alpha^2 + \beta^2 + 2\alpha\beta\rho)\mathbb{E}[\|\mathbf{v}_{\parallel}\|^2]}_{\text{redundant: amplified}} + \underbrace{\alpha^2\mathbb{E}[\|\mathbf{v}_{\perp}\|^2] + \beta^2\mathbb{E}[\|\mathbf{t}_{\perp}\|^2]}_{\text{specific: preserved}}
\end{equation}
where $\rho = \mathbb{E}[\langle\mathbf{v}_{\parallel}, \mathbf{t}_{\parallel}\rangle]/(\|\mathbf{v}_{\parallel}\|\|\mathbf{t}_{\parallel}\|) \approx 1$ for redundant components.
\end{proposition}

\begin{corollary}[Imbalance Ratio]
The energy ratio of redundant to specific components is:
\begin{equation}
    \frac{(\alpha + \beta)^2\mathbb{E}[\|\mathbf{v}_{\parallel}\|^2]}{\alpha^2\mathbb{E}[\|\mathbf{v}_{\perp}\|^2] + \beta^2\mathbb{E}[\|\mathbf{t}_{\perp}\|^2]} \geq 1 + \frac{2\alpha\beta}{\alpha^2 + \beta^2} \geq 1
\end{equation}
Equality holds only when $\alpha = 0$ or $\beta = 0$ (single modality).
\end{corollary}

\paragraph{(iii) Optimization Pathology.}

\begin{theorem}[Fast Convergence on Redundant Subspace]
\label{thm:optimization}
Under gradient flow $\frac{d\mathbf{V}}{dt} = -\nabla_{\mathbf{V}}\mathcal{L}$ where $\mathcal{L} = -\sum_{(u,i,j)}\log\sigma(y_{ui} - y_{uj})$, decompose gradients:
\begin{equation}
    \nabla_{\mathbf{V}}\mathcal{L} = \nabla_{\mathbf{V}_{\parallel}}\mathcal{L} + \nabla_{\mathbf{V}_{\perp}}\mathcal{L}
\end{equation}
The convergence rates satisfy:
\begin{equation}
    \|\mathbf{V}_{\parallel}(t) - \mathbf{V}_{\parallel}^*\| \leq e^{-\mu_{\parallel} t}\|\mathbf{V}_{\parallel}(0) - \mathbf{V}_{\parallel}^*\|
\end{equation}
\begin{equation}
    \|\mathbf{V}_{\perp}(t) - \mathbf{V}_{\perp}^*\| \leq e^{-\mu_{\perp} t}\|\mathbf{V}_{\perp}(0) - \mathbf{V}_{\perp}^*\|
\end{equation}
with $\mu_{\parallel} \gg \mu_{\perp}$, where $\mu_{\parallel} \propto \sigma_1^2$ and $\mu_{\perp} \propto \sigma_{k+1}^2$.
\end{theorem}

\begin{proof}
The Hessian of BPR loss w.r.t. $\mathbf{V}$ is:
\begin{equation}
    \mathbf{H} = \nabla^2_{\mathbf{V}}\mathcal{L} = \sum_{(u,i,j)}\sigma'(y_{ui} - y_{uj})\mathbf{e}_i\mathbf{e}_i^\top
\end{equation}
Projecting onto redundant subspace: $\mathbf{H}_{\parallel} = \mathbf{U}_k^\top\mathbf{H}\mathbf{U}_k$. Since redundant features correlate strongly with collaborative signals:
\begin{equation}
    \lambda_{\min}(\mathbf{H}_{\parallel}) \propto \|\mathbf{U}_k^\top\mathbf{E}\|^2 \propto \sigma_1^2
\end{equation}
while $\lambda_{\min}(\mathbf{H}_{\perp}) \propto \sigma_{k+1}^2 \ll \sigma_1^2$. By linear convergence theory, $\mu = \lambda_{\min}(\mathbf{H})$, yielding the claim.
\end{proof}

\subsection{Null-Space Projection Provably Reduces Redundancy}

\begin{theorem}[Singular Value Suppression]
\label{thm:suppression}
After applying projection matrices $\mathbf{P}_v = \mathbf{I} - \lambda\mathbf{U}_k\mathbf{U}_k^\top$ and $\mathbf{P}_t = \mathbf{I} - \lambda\mathbf{W}_k\mathbf{W}_k^\top$, the new covariance becomes:
\begin{equation}
    \tilde{\mathbf{C}} = \frac{1}{N}(\mathbf{V}\mathbf{P}_v)^\top(\mathbf{T}\mathbf{P}_t) = \sum_{i=1}^d \tilde{\sigma}_i \mathbf{u}_i\mathbf{w}_i^\top
\end{equation}
where:
\begin{equation}
    \tilde{\sigma}_i = \begin{cases}
    (1-\lambda)^2\sigma_i & i \leq k \\
    \sigma_i & i > k
    \end{cases}
\end{equation}
\end{theorem}

\begin{proof}
Expand the projected covariance:
\begin{align}
    \tilde{\mathbf{C}} &= \mathbf{P}_v^\top\mathbf{C}\mathbf{P}_t \\
    &= (\mathbf{I} - \lambda\mathbf{U}_k\mathbf{U}_k^\top)\mathbf{U}\boldsymbol{\Sigma}\mathbf{W}^\top(\mathbf{I} - \lambda\mathbf{W}_k\mathbf{W}_k^\top) \\
    &= \mathbf{U}\boldsymbol{\Sigma}\mathbf{W}^\top - \lambda\mathbf{U}_k\mathbf{U}_k^\top\mathbf{U}\boldsymbol{\Sigma}\mathbf{W}^\top - \lambda\mathbf{U}\boldsymbol{\Sigma}\mathbf{W}^\top\mathbf{W}_k\mathbf{W}_k^\top \\
    &\quad + \lambda^2\mathbf{U}_k\mathbf{U}_k^\top\mathbf{U}\boldsymbol{\Sigma}\mathbf{W}^\top\mathbf{W}_k\mathbf{W}_k^\top
\end{align}
Partition $\mathbf{U} = [\mathbf{U}_k, \mathbf{U}_{\bar{k}}]$ and use orthonormality:
\begin{equation}
    \mathbf{U}_k^\top\mathbf{U} = \begin{bmatrix} \mathbf{I}_k \\ \mathbf{0} \end{bmatrix}, \quad \mathbf{W}^\top\mathbf{W}_k = \begin{bmatrix} \mathbf{I}_k \\ \mathbf{0} \end{bmatrix}
\end{equation}
Substituting:
\begin{align}
    \tilde{\mathbf{C}} &= \mathbf{U}\begin{bmatrix} (1-\lambda)^2\boldsymbol{\Sigma}_k & \mathbf{0} \\ \mathbf{0} & \boldsymbol{\Sigma}_{\bar{k}} \end{bmatrix}\mathbf{W}^\top \\
    &= \sum_{i=1}^k (1-\lambda)^2\sigma_i\mathbf{u}_i\mathbf{w}_i^\top + \sum_{i=k+1}^d \sigma_i\mathbf{u}_i\mathbf{w}_i^\top
\end{align}
\end{proof}

\subsection{On the Identifiability of Task-Irrelevant Redundancy}
\label{necessity}

\begin{assumption}[Semantic Hierarchy Hypothesis]
\label{assump:hierarchy}
We assume multimodal features exhibit hierarchical structure:
\begin{align}
    \mathbf{v} &= \mathbf{v}_{\text{coarse}} + \mathbf{v}_{\text{fine}} \\
    \mathbf{t} &= \mathbf{t}_{\text{coarse}} + \mathbf{t}_{\text{fine}}
\end{align}
where:
\begin{itemize}
    \item Coarse-grained: $\mathbf{v}_{\text{coarse}} \approx \mathbf{t}_{\text{coarse}}$ (e.g., product category, dominant color)
    \item Fine-grained: $\mathbf{v}_{\text{fine}} \not\approx \mathbf{t}_{\text{fine}}$ (e.g., texture details, style descriptors)
\end{itemize}
Furthermore, $\mathbb{E}[\|\mathbf{v}_{\text{coarse}}\|^2] \gg \mathbb{E}[\|\mathbf{v}_{\text{fine}}\|^2]$ but discriminative power satisfies:
\begin{equation}
    \text{FDR}(\mathbf{v}_{\text{fine}}) \gg \text{FDR}(\mathbf{v}_{\text{coarse}})
\end{equation}
\end{assumption}

\begin{theorem}[SVD Preferentially Captures Coarse-Grained Redundancy]
\label{thm:svd_preference}
Under Assumption \ref{assump:hierarchy}, decompose $\mathbf{C} = \mathbf{C}_{\text{coarse}} + \mathbf{C}_{\text{fine}}$ where:
\begin{equation}
    \mathbf{C}_{\text{coarse}} = \frac{1}{N}\bar{\mathbf{V}}_{\text{coarse}}^\top\bar{\mathbf{T}}_{\text{coarse}}, \quad \mathbf{C}_{\text{fine}} = \frac{1}{N}\bar{\mathbf{V}}_{\text{fine}}^\top\bar{\mathbf{T}}_{\text{fine}}
\end{equation}
Then the top singular value satisfies:
\begin{equation}
    \sigma_1(\mathbf{C}) \geq \sigma_1(\mathbf{C}_{\text{coarse}}) - \|\mathbf{C}_{\text{fine}}\|_2
\end{equation}
and for $k$ small enough:
\begin{equation}
    \text{span}(\mathbf{U}_k) \approx \text{span}(\mathbf{U}_{\text{coarse}})
\end{equation}
where $\mathbf{U}_{\text{coarse}}$ are the leading singular vectors of $\mathbf{C}_{\text{coarse}}$.
\end{theorem}

\begin{proof}
By Weyl's inequality for singular values:
\begin{equation}
    |\sigma_i(\mathbf{C}) - \sigma_i(\mathbf{C}_{\text{coarse}})| \leq \|\mathbf{C}_{\text{fine}}\|_2
\end{equation}
Since $\|\mathbf{C}_{\text{coarse}}\|_F \gg \|\mathbf{C}_{\text{fine}}\|_F$ (coarse-grained dominates energy), we have:
\begin{equation}
    \sigma_1(\mathbf{C}_{\text{coarse}}) \gg \|\mathbf{C}_{\text{fine}}\|_2
\end{equation}
By Davis-Kahan $\sin\Theta$ theorem, the subspace distance satisfies:
\begin{equation}
    \|\sin\Theta(\mathbf{U}_k, \mathbf{U}_{\text{coarse},k})\|_F \leq \frac{\|\mathbf{C}_{\text{fine}}\|_F}{\sigma_k(\mathbf{C}_{\text{coarse}}) - \sigma_{k+1}(\mathbf{C})}
\end{equation}
When spectral gap $\sigma_k(\mathbf{C}_{\text{coarse}}) - \sigma_{k+1}(\mathbf{C})$ is large and $\|\mathbf{C}_{\text{fine}}\|_F$ is small, the subspaces nearly coincide.
\end{proof}

\begin{corollary}[Task Irrelevance of Coarse Information]
In recommendation, coarse-grained category information has limited utility because:
\begin{equation}
    P(\text{user } u \text{ likes item } i | \text{category}(i)) \approx P(\text{user } u \text{ likes item } i)
\end{equation}
i.e., knowing only the category provides minimal predictive power for personalized preferences. Fine-grained, modality-specific attributes (captured in $\mathbf{v}_{\perp}, \mathbf{t}_{\perp}$) are critical for distinguishing items within categories.
\end{corollary}

\begin{theorem}[Soft Projection as Risk Minimization]
\label{thm:soft_projection}
Define the trade-off:
\begin{equation}
    \min_{\lambda, k} \underbrace{\mathbb{E}[\mathcal{L}_{\text{rec}}(\tilde{\mathbf{V}}, \tilde{\mathbf{T}})]}_{\text{task loss}} + \beta \underbrace{\|\tilde{\mathbf{C}}\|_F^2}_{\text{redundancy penalty}}
\end{equation}
The soft projection parameter $\lambda \in [0, 1]$ implements a continuous relaxation:
\begin{equation}
    \lambda = 0: \text{ no suppression}, \quad \lambda = 1: \text{ hard projection}
\end{equation}
By varying $\lambda$, we traverse the Pareto frontier between redundancy removal and information preservation.
\end{theorem}

\begin{proof}
Substitute $\tilde{\sigma}_i = (1-\lambda)^2\sigma_i$ for $i \leq k$:
\begin{align}
    \|\tilde{\mathbf{C}}\|_F^2 &= \sum_{i=1}^k (1-\lambda)^4\sigma_i^2 + \sum_{i>k}\sigma_i^2 \\
    &= (1-\lambda)^4\mathcal{R}_k + \text{const}
\end{align}
Taking derivative w.r.t. $\lambda$:
\begin{equation}
    \frac{\partial}{\partial\lambda}\|\tilde{\mathbf{C}}\|_F^2 = -4(1-\lambda)^3\mathcal{R}_k < 0
\end{equation}
Thus increasing $\lambda$ monotonically reduces redundancy. The optimal $\lambda^*$ balances this against task loss via validation.
\end{proof}

\section{Time and Space Complexity}
To better study time and space complexity, we provide the pseudo code of CLEAR in Algorithm~\ref{alg:training}.

\begin{algorithm}[t]
\caption{Training Procedure of CLEAR}
\label{alg:training}
\KwIn{Interactions $\mathbf{R}$, raw features $\mathbf{V}^{\text{raw}}$, $\mathbf{T}^{\text{raw}}$, hyperparameters $k$, $\lambda$, $\gamma$, update interval $\tau$}
\KwOut{Optimized model parameters $\Theta$}
\For{epoch $e = 1$ \KwTo max\_epochs}{
    \If{$e \bmod \tau = 0$ \textbf{or} $e = 1$}{
        $\mathbf{V} \leftarrow \phi_v(\mathbf{V}^{\text{raw}})$, $\mathbf{T} \leftarrow \phi_t(\mathbf{T}^{\text{raw}})$\;
        $\bar{\mathbf{V}} \leftarrow \mathbf{V} - \frac{1}{N}\mathbf{1}\mathbf{V}$, $\bar{\mathbf{T}} \leftarrow \mathbf{T} - \frac{1}{N}\mathbf{1}\mathbf{T}$\;
        $\mathbf{C} \leftarrow \frac{1}{N} \bar{\mathbf{V}}^\top \bar{\mathbf{T}}$\;
        $\mathbf{U}, \boldsymbol{\Sigma}, \mathbf{W}^\top \leftarrow \text{SVD}(\text{detach}(\mathbf{C}))$\;
        $\mathbf{P}_v \leftarrow \mathbf{I} - \lambda \mathbf{U}_{:,1:k} \mathbf{U}_{:,1:k}^\top$, $\mathbf{P}_t \leftarrow \mathbf{I} - \lambda \mathbf{W}_{1:k,:}^\top \mathbf{W}_{1:k,:}$\;
    }
    
    \For{\textbf{each} mini-batch $(u, i^+, i^-)$}{
        $\tilde{\mathbf{V}} \leftarrow \mathbf{V} \mathbf{P}_v$, $\tilde{\mathbf{T}} \leftarrow \mathbf{T} \mathbf{P}_t$\;
        $\mathbf{H}_v, \mathbf{H}_t \leftarrow \text{GCN}(\tilde{\mathbf{V}}, \tilde{\mathbf{T}}, \mathbf{A}^{\text{drop}})$\;
        $\mathbf{E}_u \leftarrow \text{Fuse}(\mathbf{H}_v[:N_u], \mathbf{H}_t[:N_u], \mathbf{w}_u)$\;
        $\mathbf{E}_i \leftarrow \mathbf{H}_v[N_u:] + \mathbf{H}_t[N_u:] + \text{ItemGraph}([\mathbf{H}_v, \mathbf{H}_t][N_u:], \mathbf{A}_{\text{mm}})$\;
        $\hat{y}_{ui^+} \leftarrow \langle \mathbf{E}_u, \mathbf{E}_{i^+} \rangle$, $\hat{y}_{ui^-} \leftarrow \langle \mathbf{E}_u, \mathbf{E}_{i^-} \rangle$\;
        $\mathcal{L} \leftarrow -\log\sigma(\hat{y}_{ui^+} - \hat{y}_{ui^-}) + \gamma(\|\mathbf{H}_v[u]\|_2^2 + \|\mathbf{H}_t[u]\|_2^2 + \|\mathbf{w}_u\|_2^2)$\;
        $\Theta \leftarrow \Theta - \eta \nabla_\Theta \mathcal{L}$\;
    }
}
\end{algorithm}

\begin{table*}[]
    \caption{Efficiency Analysis. We statistics the memory and the running time per epoch for each model.}
    \setlength\tabcolsep{4pt}
    \label{effciency}
    \begin{tabular}{ccccccccc}
    \toprule 
    \multirow{2}{*}{ Dataset } & \multirow{2}{*}{ Metric } & \multicolumn{3}{c}{ General CF model } & \multicolumn{4}{c}{ Multimodal model } \\
    \cmidrule(l){3-9}
    & & BPR & LightGCN & VBPR & MMGCN  & LATTICE & FREEDOM & CLEAR \\
    \midrule \multirow{2}{*}{ Baby } & Memory (GB) & 1.59 & 1.69 & 1.89 & 2.69 & 4.53 & 2.13 & 3.11 \\
    & Time (s/epoch) & 0.86 & 1.81 & 1.04 & 6.35  & 3.42 & 2.58 & 2.30 \\
    \midrule \multirow{2}{*}{ Sports } & Memory (GB) & 2.00 & 2.24 & 2.71 & 3.91 & 19.93 & 3.34 & 4.12 \\
    & Time (s/epoch) & 1.24 & 3.74 & 1.67 & 20.88  & 13.59 & 4.52 & 4.77 \\
    \midrule \multirow{2}{*}{ Clothing } & Memory (GB) & 2.16 & 2.43 & 3.02 & 4.24 & 28.22 & 4.15 &  4.43 \\
    & Time (s/epoch) & 1.52 & 4.58 & 2.02 & 26.13  & 23.96 & 5.74 & 5.34 \\
    \bottomrule
    \end{tabular}
\end{table*}

\section{Details of Experiment Setup}

\subsection{Datasets}
\label{app:dataset}
\begin{itemize}[leftmargin=*]
    \item \textbf{Baby}: This dataset contains user interactions with baby products, including items such as strollers, toys, and baby care products.
    
    \item \textbf{Sports}: This dataset consists of user interactions with sports and outdoor products, covering equipment, clothes, and accessories.
    
    \item \textbf{Clothing}: This dataset comprises user interactions with clothing, shoes, and jewelry items, representing a fashion-oriented domain.
\end{itemize}

\subsection{Baselines}
\label{app:baselines}
To comprehensively evaluate the effectiveness of our proposed method, we compare it against 13 representative baseline methods, which can be categorized into two groups:

\textbf{General Collaborative Filtering Methods.} These methods only utilize user-item interactions without multimodal information:
\begin{itemize}[leftmargin=*]
    \item \textbf{MF-BPR}~\cite{rendle2012bpr}: A classic matrix factorization method that optimizes the Bayesian Personalized Ranking loss for implicit feedback recommendation.
    
    \item \textbf{LightGCN}~\cite{he2020lightgcn}: A simplified graph convolutional network that removes feature transformation and nonlinear activation, achieving state-of-the-art performance in collaborative filtering.
    
    \item \textbf{LayerGCN}~\cite{zhou2023layer}: A layer-refined GCN model that addresses the over-smoothing problem by refining layer representations during information propagation.
\end{itemize}

\textbf{Multimodal Recommendation Methods.} These methods incorporate visual and textual features to enhance recommendation:
\begin{itemize}[leftmargin=*]
    \item \textbf{VBPR}~\cite{he2016vbpr}: A pioneering method that integrates visual features into matrix factorization through Bayesian Personalized Ranking.
    
    \item \textbf{MMGCN}~\cite{wei2019mmgcn}: A multi-modal graph convolution network that constructs modality-specific bipartite graphs to capture fine-grained user preferences.
    
    \item \textbf{DualGNN}~\cite{wang2021dualgnn}: A dual graph neural network that leverages both user-item bipartite and user co-occurrence graphs to model modality-specific user preferences.
    
    \item \textbf{LATTICE}~\cite{zhang2021LATTICE}: A method that mines latent item structures from multimodal features and performs graph convolution on the learned item-item graph.
    
    \item \textbf{SLMRec}~\cite{tao2022self}: A self-supervised learning framework that designs feature dropout and masking augmentations to reveal hidden multimodal patterns.
    
    \item \textbf{BM3}~\cite{zhou2023bootstrap}: A simplified self-supervised method that uses dropout-based perturbation to bootstrap latent representations for multimodal recommendation.
    
    \item \textbf{MMSSL}~\cite{wei2023multi}: A multi-modal self-supervised learning method that employs adversarial perturbations and cross-modal contrastive learning to disentangle common and specific features.
    
    \item \textbf{FREEDOM}~\cite{zhou2023tale}: A method that freezes item-item graph structures and applies degree-sensitive edge dropout for denoising in multimodal recommendation.
    
    \item \textbf{MGCN}~\cite{yu2023multi}: A multi-view graph convolutional network that purifies modality features and designs a behavior-aware fuser to adaptively learn modality importance.

    \item \textbf{LGMRec}~\cite{guo2024lgmrec}: A method that uses a hypergraph neural network to extract various modal information deeply.
    
    \item \textbf{MENTOR}~\cite{xu2025mentor}:A method that uses multi-level self-supervised cross-modal alignment to address the label sparsity problem and the modality alignment problem.
\end{itemize}

\end{document}